\documentclass[twoside,leqno]{article}
\newif\ifsubmission
\submissiontrue
\submissionfalse %

\newif\ifcomment
\commenttrue

\newif\ifcamera
\cameratrue
\camerafalse %

\ifsubmission
\commentfalse
\author{Anonymous Authors}
\else
\author{
  Yotam Kenneth-Mordoch
  \qquad
  Robert Krauthgamer%
  \thanks{ The Harry Weinrebe Professorial Chair of Computer Science.
    Work partially supported by the Israel Science Foundation grant \#1336/23,
    by the Israeli Council for Higher Education (CHE) via the Weizmann Data Science Research Center,
    and by a research grant from the Estate of Harry Schutzman.
  }
  \\ Weizmann Institute of Science
  \\ \texttt{\{yotam.kenneth,robert.krauthgamer\}@weizmann.ac.il}
}

\fi

\usepackage[letterpaper, margin=1in]{geometry}

\ifcamera
\usepackage{siamproceedings}
\usepackage{amsmath}

\else
\usepackage{amsmath}
\usepackage{amsthm}
\usepackage[hypertexnames=false]{hyperref}
\usepackage[capitalise]{cleveref} 

\newtheorem{theorem}{Theorem}[section]

\newtheorem{lemma}[theorem]{Lemma}

\newtheorem{definition}[theorem]{Definition}

\hypersetup{colorlinks=true,allcolors=blue}

\fi

\usepackage{algorithm}
\usepackage[noend]{algpseudocode}

\usepackage{amsfonts}
\usepackage{import} 
\usepackage{xifthen}
\usepackage{mathtools}
\usepackage{arydshln}%
\usepackage{dsfont}
\usepackage[cal=esstix]{mathalfa}
\usepackage{thm-restate}
\usepackage{pgf}
\usepackage{comment}
\usepackage[colorinlistoftodos,textsize=tiny,textwidth=2cm,color=red!25!white,obeyFinal]{todonotes}
\newtheorem{claim}[theorem]{Claim}

\ifdefined\AddToHook
\AddToHook{cmd/appendix/before}{\crefalias{section}{appendix}}
\fi

\newcommand{\Exp}[2]{\mathbb{E}_{ #2}\left[ #1 \right]}

\newcommand{\Probability}[1]{\Pr\left[ #1 \right]}

\newcommand{\tO}{\tilde{O}}

\DeclareMathOperator{\polylog}{polylog}
\providecommand{\set}[1]{{\{#1\}}}
\newcommand{\eqdef}{\coloneqq}

\ifcomment

\else

\fi

\title{Simple Algorithms for Fully Dynamic Edge Connectivity}
\begin{document}

\maketitle

\begin{abstract}
  In the fully dynamic edge connectivity problem,
  the input is a simple graph $G$ undergoing edge insertions and deletions,
  and the goal is to maintain its edge connectivity, denoted $\lambda_G$.
  We present two simple randomized algorithms solving this problem.
  The first algorithm maintains the edge connectivity in worst-case update time $\tilde{O}(n)$ per edge update,
  matching the known bound but with simpler analysis.
  Our second algorithm achieves worst-case update time $\tilde{O}(n/\lambda_G)$
  and worst-case query time $\tilde{O}(n^2/\lambda_G^2)$,
  which is the first algorithm with worst-case update and query time $o(n)$
  for large edge connectivity, namely, $\lambda_G = \omega(\sqrt{n})$.
\end{abstract}

\section{Introduction}
\label{sec:introduction}
Finding the edge connectivity of a graph is a fundamental algorithmic problem in graph theory.
The \emph{edge connectivity} of a graph $G=(V,E)$, denoted $\lambda_G$, is defined as the minimum number of edges that need to be removed from $G$ in order to disconnect it;
this is equivalent to the value of a global minimum cut in $G$.
It has long been known that it is possible to find the edge connectivity of a graph in $\tO(m)$ time, where $n=|V|,m=|E|$ and $\tO(\cdot)$ hides polylogarithmic factors in $n$ \cite{Karger00,KT19,HRW20,MN20,GMW20}.

We consider the problem of maintaining the edge connectivity of a graph $G=(V,E)$ in a fully dynamic setting, where the input is a \emph{dynamic graph} presented as a fixed set of vertices and a sequence of edge insertions and deletions, and the goal is to report the edge connectivity after each update. 
This problem has a rich history: it was originally studied for constant values of edge connectivity (namely $\lambda_G\in \{1,2,3,4\}$) see for example \cite{EGIN97,HLT01,KKM13,CGLNPS20};
for the case $\lambda_G=1$ there exists a randomized algorithm using worst-case update time $\polylog (n)$ and a deterministic algorithm using amortized update time $\polylog (n)$ \cite{HLT01,KKM13}.
These results are complemented by a deterministic algorithm for the case $\lambda_G = \log^{o(1)}n$, which achieves worst-case update time $n^{o(1)}$ \cite{JST24}.

For the case of general $\lambda_G$, there exist deterministic algorithms that maintain the edge connectivity in worst-case time $\tO(\min\left(\lambda_{\max}^{5.5}\sqrt{n},m^{1-1/12},m^{11/13}n^{1/13},n^{3/2} \right))$ per update, where $\lambda_{\max}$ is the maximum edge connectivity of $G$ throughout the algorithm's execution \cite{Thorup07,GHNSTW23,VC25}.
Allowing for randomization, there exists a randomized algorithm that maintains the edge connectivity in worst-case time $\tO(n)$ per update \cite{GHNSTW23}.
Finally, for approximate maintenance which asks to maintain a $(1+\epsilon)$-approximation of the edge connectivity \cite{TK00,Thorup07,EHL25}, one can achieve a $(1+o(1))$-approximation with worst-case update time $n^{o(1)}$ \cite{EHL25}. 

\subsection{Contributions}
We focus on general edge connectivity, and present two new randomized algorithms that achieve worst-case update time $\tO(n)$ and $\tO(n^2/\lambda_G^2)$ respectively.
Both algorithms are based on the same high-level idea, which is an efficient and simple implementation of a contraction algorithm that preserves the edge connectivity.
Our first algorithm matches the result of \cite{GHNSTW23} but has a simpler and more intuitive proof, while our second algorithm offers improved performance for dense graphs.
Throughout, update time (of an algorithm) refers to the worst-case time per edge update unless stated otherwise.
We say that a randomized algorithm succeeds with high probability if its probability of success is at least $1-1/n^c$ for any constant $c>0$.
\begin{theorem}
    \label{theorem:main-matching-GHNSTW23}
    There exists a fully dynamic algorithm that, given an unweighted dynamic graph $G$ on $n$ vertices, maintains the edge connectivity of $G$ with update time $\tO(n)$.
    The algorithm is randomized and succeeds with high probability.
\end{theorem}
Our second algorithm improves on the result of \cite{GHNSTW23} for large edge connectivity, namely when $\lambda_G = \tilde{\omega}(\sqrt{n})$.
Our result is actually stronger, since it is parametrized by the minimum degree of the graph at the time of the update, $\delta_G$, which is only larger than the edge connectivity $\lambda_G$ and hence it implies an update time of $\tO(n^2/\lambda^2_G)$.
Therefore, \Cref{theorem:main-better-for-larger} improves on \cite{GHNSTW23} for large edge connectivity, i.e. $\lambda_G \geq \tilde{\omega}(\sqrt{n})$, and is in fact the first algorithm to achieve worst-case update time $o(n)$ for large edge connectivity.
Moreover, the improvement applies also to graphs with large minimum degree but small edge connectivity.
The improvement is even larger when separating the update and query times, i.e. if the algorithm only needs to report the edge connectivity occasionally.
\begin{theorem}
    \label{theorem:main-better-for-larger}
    There exists a fully dynamic algorithm that, given an unweighted dynamic graph $G$ on $n$ vertices, maintains the edge connectivity of $G$ with update time $\tO(n/\delta_G)$ and query time $\tO(n^2/\delta_G^2)$.
    The algorithm is randomized and succeeds with high probability.
\end{theorem}
We now provide several details about the above results.
The first one is that both results can report not only the exact edge connectivity, but also the edges that form a minimum cut, which incurs an additional $\tO(\lambda_G)$ time per query.
When reporting only the value of the edge connectivity, the algorithms are adversarially-robust, i.e. they work against an adaptive adversary.
This is because the edge connectivity is a deterministic function of the graph and hence, as long as the algorithm returns the correct result, reporting it does not leak any information to the adversary.
In contrast, when reporting the edges of the minimum cut, our algorithms only work against an oblivious adversary.
Since in both cases the adversary lacks knowledge of the algorithms' internal state, our analysis only considers the case of an oblivious adversary.

The second detail is that it is possible to combine the results, to obtain an algorithm with update time $\tO(\min\set{n, n^2/\delta_G^2})$.
A comparison of existing results with our work is given in Figure~\ref{figure:summary-of-existing-results}.
As explained above we use the fact that the update time of \Cref{theorem:main-better-for-larger} can be stated as $\tO(n^2/\lambda_G^2)$.

\begin{figure}
    \label{figure:summary-of-existing-results}
    \centering
    \input{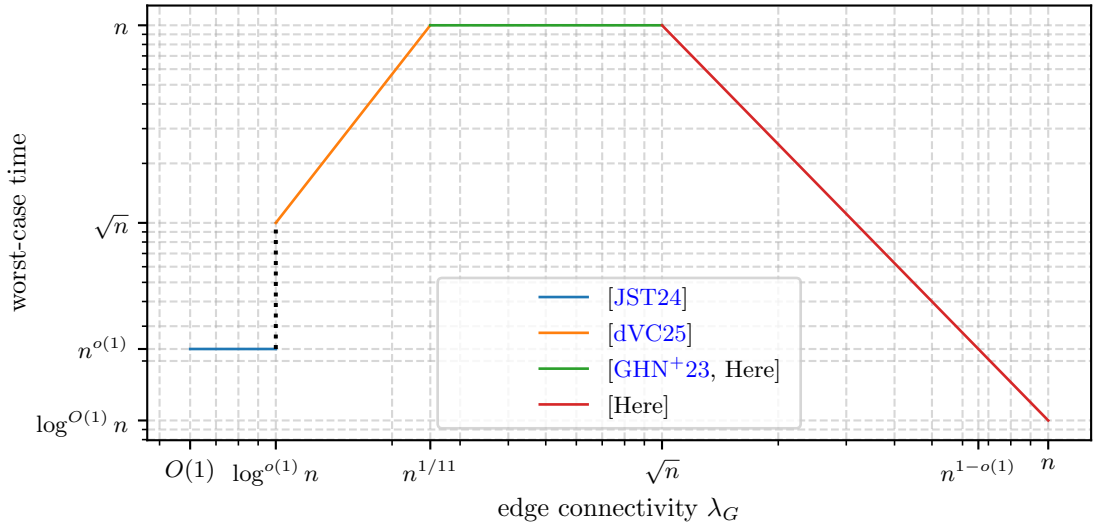}
    \caption{Schematic comparison of existing algorithms for dynamic edge connectivity.}
\end{figure}

\subsection{Future Work}
The main open question remaining is to achieve worst-case update time $o(n)$ for all edge connectivity values.
Prior work achieves this for edge connectivity $\lambda_G = o(n^{1/11}/\polylog n)$ \cite{Thorup07,JST24,VC25}.
Our results complement these results and achieve it for large edge connectivity, namely $\lambda_G = \tilde{\omega}(\sqrt{n})$.
Perhaps surprisingly, the remaining case is that of intermediate connectivity, $\lambda_G\in [n^{1/11},n^{1/2}]$.

\section{Proof of Main Theorems}
\label{sec:technical-overview}
A powerful tool for finding a minimum cut that has been used extensively in past work is a family of contraction procedures first introduced by Kawarabayashi and Thorup \cite{KT19}, which we call \emph{KT-style contractions} \cite{KT19,HRW20,GNT20,Saranurak21,AEGLMN22,KK25}.
In general, KT-style contractions preserve the minimum cut value of an input graph $G$ while reducing the number of vertices to $\tO(n/\delta_G)$, where $\delta_G$ is the minimum degree in $G$.
The edge connectivity can then be more efficiently computed on a smaller contracted graph $G'=(V',E')$.
This general approach has also been employed to design an efficient algorithm for dynamic edge connectivity \cite{GHNSTW23}.
Specifically, the algorithm leverages the $2$-out contraction technique of \cite{GNT20} to maintain a contracted multigraph $G'$ of the input graph $G$ \cite{GHNSTW23}, and constructs a maximal $\delta_G$-packing of forests in $G'$.

\begin{definition}[Maximal $k$-Packing of Forests]
    Given a (multi)graph $G$, a maximal $k$-packing of forests is a set of edge-disjoint forests $T_1,\ldots,T_k$ of $G$ that are maximal, i.e. for every $i\in [k]$ and $e\in E\setminus \cup_j T_j$, the edge set $T_i \cup \{e\}$ contains a cycle.
    A maximal $k$-packing of forests in a multigraph includes at most one parallel edge in each forest.
\end{definition}

The main idea of \cite{GHNSTW23} is that the edge connectivity of $G$ is at most $\delta_G$, and therefore if a global minimum cut of $G$ is preserved in $G'$, then all its edges are contained in a maximal $\delta_G$-packing of forests in the multigraph $G'$.
Therefore, the edge connectivity of $G'$ (and hence of $G$) is equal to the edge connectivity of $H=(V',T_1\cup\ldots\cup T_{\delta_G})$, where $\set{T_i}_{i=1}^{\delta_G}$ is a maximal $\delta_G$-packing of forests in $G'$.
Furthermore, the number of edges in $H$ is at most $\tO(n)$ since each forest $T_i$ has at most $O(|V'|)=\tO(n/\delta_G)$ edges, and hence the edge connectivity of $H$ can be computed in time $\tO(n)$ using some static algorithm for edge connectivity.

From a simplicity perspective, the main drawback of \cite{GHNSTW23} is that it does not maintain an explicit contraction of the graph, but rather only its connected components.
It then recovers the forest packing using linear-sketching techniques, which complicates the algorithm.
Another drawback of this approach is that linear sketching only recovers unweighted edges of the contracted multigraph $G'$.
However, by reinterpreting the contracted multigraph as a weighted graph (merging parallel edges), we can take advantage of the fact that when $\delta_G \ge \omega(\sqrt{n})$ the graph $G'$ becomes very small.%
\footnote{Such a reinterpretation is natural and was used before in the context of KT-style contractions, e.g. \cite{KK25}.}
In fact, the number of edges in $G'$ is $\tO (n^2/\delta_G^2)$, which enables solving the edge connectivity problem in $\tilde{o}(n)$ worst-case time per edge update when $\delta_G \ge \omega(\sqrt{n})$.

Our algorithms follow the general recipe of \cite{GHNSTW23}, with the main difference being that our contraction algorithm maintains an explicit KT-style contraction as a weighted graph.%
\footnote{To recover the edges of a minimum cut we can maintain a list of the edges of $G$ that are mapped to each edge of $G'$ and report the edges corresponding to the minimum cut in $G'$.}
This both simplifies the algorithm and improves the update time whenever $\delta_G \ge \omega(\sqrt{n})$.
Specifically, we leverage the $\tau$-star contraction procedure of \cite{AEGLMN22,KK25}, which was originally designed for finding a global minimum cut in the cut-query model.
The procedure works as follows.
Sample a set of \emph{center vertices} $R\subseteq V$ by including each vertex $v\in V$ independently with probability $p=O(\log n/\tau)$.
Let $H=\set{v\in V\setminus R \mid d_G(v)\ge \tau}$ be the set of non-center vertices with degree at least $\tau$,
where throughout $d_G(v)$ is the degree of $v$ in $G$.
Then, for every $v\in H$, uniformly sample a vertex $r\in N_G(v)\cap R$, where $N_G(v)$ is the neighborhood of $v$ in $G$, and contract the corresponding edge $(v,r)$, keeping parallel edges, which yields a contracted multigraph $G'$.
The main guarantee of this procedure is that if $G$ has a minimum cut that is non-trivial, i.e., not composed of a single vertex, then $\lambda_G=\lambda_{G'}$ with constant probability.
An illustration of the $\tau$-star contraction is given in \Cref{fig:star-contraction}.
The following theorem states this formally.
\begin{theorem}[\cite{AEGLMN22,KK25}]
    \label{theorem:original-star-contraction}
    Let $G=(V,E)$ be an unweighted graph on $n$ vertices with at least one minimum cut that is non-trivial.
    Then, $\tau$-star-contraction with $p = 800\log n/\tau$ yields a contracted graph $G'$ such that,
    \begin{enumerate}
        \item if $\tau \le \delta_G$, then with probability at least $1-1/n^4$ all vertices in $V\setminus R$ are contracted and $G'$ has at most $\tO(n/\tau)$ vertices, and
        \item $\lambda_{G'}\ge \lambda_G$, and equality holds with probability at least $2\cdot 3^{-13}$.
    \end{enumerate}
\end{theorem}
\begin{figure}[t]
    \centering
    \includegraphics[width=0.95\textwidth]{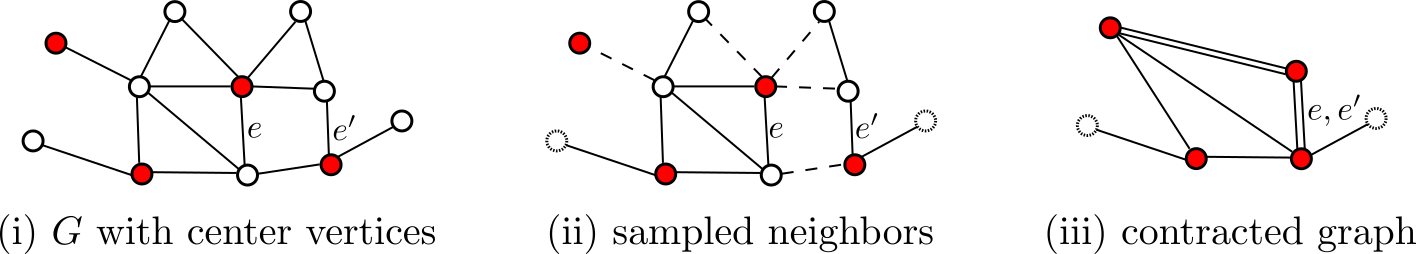}
    \caption{
    An illustration of the $\tau$-star contraction procedure with $\tau=2$.
    (i) A random set of center vertices $R$ is sampled (indicated by red).
    (ii) Each non-center vertex in $H=\{v\in V\setminus R \mid d_G(v)\ge \tau\}$ chooses a random neighbor in $R$ (indicated by dashed edges); 
    vertices not in $H\cup R$ are indicated by a striped pattern.
    (iii) The corresponding edges are contracted (keeping parallel edges) to obtain a contracted graph $G'$.
    }
    \label{fig:star-contraction}
\end{figure}
The following technical lemma provides the guarantees of the $\tau$-star contraction procedure in a fully dynamic setting.
Using this lemma our algorithms maintain a KT-style contraction of the graph $G$ as a weighted graph (by merging parallel edges) in a fully dynamic setting.
We call a contraction $G'$ of $G$ \emph{complete}, if it maps every edge of $G$ to an edge in $G'$, and otherwise it is \emph{incomplete}.%
\footnote{Notice that a incomplete contraction is not a contraction in the strict sense, since it does not map every edge of $G$ to an edge in $G'$.}
Our contraction procedure only tracks edges that are internal to the center vertices $R$, hence if a vertex $v\in V\setminus R$ is not contracted, then its incident edges are not mapped in $G'$.
Therefore, its contraction is complete only when all vertices in $V\setminus R$ are contracted.
Fortunately, this happens with high probability whenever $\delta_G \ge \tau$ by \Cref{theorem:original-star-contraction}.
\begin{lemma}
    \label{lemma:dynamic-star-contraction}
    There is a fully-dynamic randomized algorithm that, given as input a dynamic graph $G$ and a parameter $\tau>0$, maintains a (possibly incomplete) contraction $G'$ of $G$ as a weighted graph.
    If $\delta_G\ge \tau$ then $G'$ is complete and constitutes a $\tau$-star contraction of $G$ with high probability.
    The algorithm runs in amortized update time $\polylog (n)$ and worst-case update time $\tO(n)$ (regardless of $\delta_G$).
    Furthermore, the algorithm updates at most $O(1)$ edges in $G'$ per edge update in $G$ in expectation and $\tO(n/\tau)$ in the worst-case.
\end{lemma}
We prove this lemma in \Cref{sec:dynamic-star-contraction} and proceed now to show that it implies our two main theorems.

\subsection{Proof of Theorem \ref{theorem:main-matching-GHNSTW23}}
We begin by showing how to maintain a maximal $k$-packing of forests in a weighted graph under edge insertions and deletions.
Note that since we interpret the contracted multigraph $G'$ as a weighted graph we need the packing to maintain edge weights.%
\footnote{One can simulate changing the weight of an edge by deleting it and inserting it with the new weight.}

\begin{lemma}
    \label{lemma:forest-packing}
    There exists a fully dynamic algorithm that maintains a maximal $k$-packing of forests of a dynamic weighted graph $G=(V,E)$ on $n$ vertices with update time $\tO(k)$.
    The algorithm is randomized and succeeds with high probability against an oblivious adversary.
\end{lemma}
To prove the lemma we follow the approach of \cite{ADKKP16}, which requires the following result of \cite{KKM13,GKKT15} on maintaining a single spanning forest.
\begin{theorem}[\cite{KKM13,GKKT15}]
    \label{thereom:tree-packing-known}
    There exists a fully dynamic algorithm that maintains a weighted spanning forest $T$ of a dynamic weighted graph $G$ with $n$ vertices, with worst-case update time $O(\log^4 n)$.
    Every edge insertion in $G$ adds at most one edge to $T$.
    Every edge deletion in $G$ removes at most one edge from $T$ and may additionally add one edge to $T$.
    The algorithm is randomized and succeeds with high probability against an oblivious adversary.
\end{theorem}
\begin{proof}[Proof of \Cref{lemma:forest-packing}]
    The algorithm maintains a collection of $k$ spanning trees $\set{T_i}_{i=1}^k$, where each $T_i$ is spanning tree of the graph $G\setminus \cup_{j<i}T_j$, that is maintained using \Cref{thereom:tree-packing-known}.
    To show that the update time is $\tO(k)$, examine some edge update $e$ and level $i$.
    If the update is an insertion, then either $e$ is added to $T_i$ or the update is propagated to level $i+1$, and in both cases at most one change is propagated to level $i+1$.
    If the update is a deletion, then either $e\not\in T_i$ and this update is propagated to level $i+1$, or $e\in T_i$ and then removing it might cause a single other edge $e'$ to be added to $T_i$ from $G\setminus T_{j\le i}$.
    In this case $e'$ is from $G\setminus T_{j\le i}$, and therefore we propagates a deletion of $e'$ to level $i+1$ (however, $T_{i+1}$ is oblivious to the deletion of $e$).
    Therefore, an edge update in $G$ causes at most $O(1)$ changes to each $T_i$ and the overall update time is $\tO(k)$.
\end{proof}
\begin{proof}[Proof of \Cref{theorem:main-matching-GHNSTW23}]
    It suffices to present an algorithm that returns a global minimum cut of $G$ with constant probability, and returns some larger cut otherwise (if the algorithm is only required to return the edge connectivity then it only returns the value of this cut).
    Running $O(\log n)$ independent copies of the algorithm and taking the minimum cut found yields a global minimum cut with high probability.
    The desired algorithm works as follows.
    Begin by running in parallel $r=O(\log n)$ instances of the $\tau$-star contraction algorithm of \Cref{lemma:dynamic-star-contraction} to obtain contracted graphs $\set{G_i'=(V_i',E_i')}_{i=1}^r$, where the $i$-th instance uses $\tau = 2^i$.
    Furthermore, for each instance $G_i'$ maintain a maximal $2^{i+1}$-forest packing $\set{T_j^i}_{j=1}^{2^{i+1}}$ of the contracted graph $G_i'$ using \Cref{lemma:forest-packing}.%
    \footnote{An additional factor $2$ in the size of the packing is needed since $\delta_G/\tau\in[1,2)$.}
    In parallel, track the minimum degree $\delta_G$ of the input graph $G$.

    Let $i^*$ be the maximal such that $\delta_G\ge 2^{i^*}$,
    and let $H=(V_{i^*}',\cup_{j=1}^{2^{i^*+1}}T_j^{i^*})$
    be the graph obtained by taking the union of the maximal $2^{i^*+1}$-forest packing of $G_{i^*}'$.
    The algorithm then finds a minimum cut of $H$ using an offline algorithm for weighted graphs in time $\tO(n)$ and, returns the minimum between the edge connectivity found and $\delta_G$ (and the associated cut).
    The correctness guarantee follows from \Cref{theorem:original-star-contraction} and since the contraction is complete by \Cref{lemma:dynamic-star-contraction} (and hence $G_{i}'$ is a $2^{i}$-star contraction of $G$) whenever $\delta_G \ge \tau$.

    To analyze the time complexity, notice that maintaining each $\tau$-star-contraction instance takes worst-case time $\tO(n)$ per update by \Cref{lemma:dynamic-star-contraction}.
    To analyze the time complexity of maintaining a maximal $2^{i+1}$-forest packing, notice that each edge update in $G$ updates at most $\tO(n/2^{i})$ edges 
    in $G_i'$ by \Cref{lemma:dynamic-star-contraction}.
    Each update in $G_i'$ then takes $\tO(2^{i+1})$ time to update the maximal $2^{i+1}$-forest packing by \Cref{lemma:forest-packing}, hence the total time per edge update in $G$ is $\tO(n)$.
    Observe that the number of edges in each $2^{i+1}$-forest packing is at most $\tO(n)$, since each forest has $O(|V_i'|)=\tO(n/2^{i})$ edges and the packing has $2^{i+1}$ forests, and hence a minimum cut can be found in time $\tO(n)$.
    Finally, observe that tracking the minimum degree $\delta_G$ in $G$ can be done using a heap data structure with worst-case update time $O(\log n)$ per edge update.
    This concludes the proof of \Cref{theorem:main-matching-GHNSTW23}.
\end{proof}

\subsection{Proof of Theorem \ref{theorem:main-better-for-larger}}
At a high level, the proof of \Cref{theorem:main-better-for-larger} is similar to that of \Cref{theorem:main-matching-GHNSTW23}.
Again the algorithm maintains in parallel $r=O(\log n)$ instances of the $\tau$-star contraction algorithm of \Cref{lemma:dynamic-star-contraction}, where the $i$-th instance $G_i'$ uses $\tau = 2^i$.
However, instead of maintaining a maximal $2^{i+1}$-forest packing
of $G_i'$, the algorithm solves the edge-connectivity problem on $G_i'$ directly.

Notice that using \Cref{lemma:dynamic-star-contraction}, the worst-case time per update in $G$ is $\tO(n)$, which is larger than $\tO(n^2/\delta_G^2)$ when $\delta_G = \omega(n^{1/2})$.
To circumvent this, we de-amortize \Cref{lemma:dynamic-star-contraction};
instead of processing immediately all $d_G(v)$ edges incident to $v$ when its chosen center vertex changes, the contraction procedure places these edges in a queue and processes $\tO(n/\delta_G)$ of them after each update in $G$.
This ensures that the worst-case time complexity per update in $G$ is $\tO(n/\delta_G)$.
However, now the contraction $G'$ may be incomplete even when $\delta_G\ge \tau$, if the queue of the contraction procedure is non-empty.
Fortunately, we can show that the probability that the queue is non-empty at any given moment in time is $O(1/\log^2 n)$ as stated next and proven in \Cref{sec:dynamic-star-contraction}.
Notice that the contraction procedure always knows whether its contraction is complete, by checking whether its queue is empty and if $\delta_G\ge \tau$, hence the algorithm can easily ignore incomplete instances.
\begin{lemma}
    \label{lemma:dynamic-lazy-star-contraction}
    There is a fully-dynamic randomized algorithm that, given as input a dynamic graph $G$ and a parameter $\tau>0$, returns a (possibly incomplete) contraction $G'$ of $G$ as a weighted graph.
    The algorithm runs in amortized update time $\polylog (n)$ and worst-case time $\tO(n)$ (regardless of $\delta_G$).
    Finally, if $\delta_G \ge \tau$ then $G'$ is complete and constitutes a $\tau$-star contraction of $G$ with probability $1-O(1/\log^2 n)$.
\end{lemma}

\begin{proof}[Proof of \Cref{theorem:main-better-for-larger}]
    Again it suffices to show an algorithm that succeeds with constant probability and returns a larger cut otherwise.
    The algorithm is as follows.
    Begin by running $r= O(\log n)$ independent instances of the $\tau$-star contraction algorithm of \Cref{lemma:dynamic-lazy-star-contraction} to obtain contracted graphs $\set{G_i'}_{i=1}^r$, where the $i$-th instance uses $\tau = 2^i$.
    In parallel, track the minimum degree $\delta_G$ of the input graph $G$.
    Let $i^*$ be the maximal such that $\delta_G \ge 2^{i^*}$.

    If the contraction $G_{i^*}'$ is incomplete, then fail (or return an arbitrary cut).
    Otherwise, find a minimum cut using an offline algorithm for edge connectivity in weighted graphs in time $\tO(n^2/\delta_G^2)$.
    Finally, return the minimum between the edge connectivity of $G_{i^*}'$ and the minimum degree $\delta_G$ (and the associated cut).
    Notice that by \Cref{lemma:dynamic-lazy-star-contraction} if $G_{i^*}'$ is a $2^{i^*}$-star contraction, the edge connectivity of $G_{i^*}'$ is equal to the edge connectivity of $G$ with constant probability.
    Furthermore, $G_{i^*}'$ is complete (and hence a $\tau$-star contraction) with probability at least $1-O(1/\log^2 n)$ by \Cref{lemma:dynamic-lazy-star-contraction}.
    Therefore, using a union bound over the two events the algorithm finds the edge connectivity with constant probability.

    To analyze the time complexity, notice that maintaining a single $\tau$-star contraction instance using \Cref{lemma:dynamic-lazy-star-contraction} takes amortized time $\polylog (n)$ and worst-case time $\tO(n/\delta_G)$ per edge update in $G$.
    Maintaining the minimum degree $\delta_G$ in $G$ can be done using a heap data structure with worst-case update time $O(\log n)$ per edge update.
    Therefore, the time per edge update in $G$ is $\tO(n/\delta_G)$.
    Finally, the query time is determined by the edge-connectivity computation in $G_{i^*}'$ which takes $\tO(n^2/2^{2i^*})=\tO(n^2/\delta_G^2)$ time.
\end{proof}

\section{Dynamic Star Contraction}
\label{sec:dynamic-star-contraction}
In this section we present a fully dynamic algorithm that maintains a $\tau$-star contraction of the graph $G$, proving \Cref{lemma:dynamic-star-contraction,lemma:dynamic-lazy-star-contraction}.
The algorithm for both results is the same, except for minor changes in the edge-update procedure which we detail in the proof of \Cref{lemma:dynamic-lazy-star-contraction}, and presented in \Cref{alg:dynamic-star-contraction}.
To maintain the $\tau$-star contraction, we use a stable dynamic uniform sampler.
The construction of this sampler uses standard techniques and is described in \Cref{lemma:stable-uniform-dynamic-sampling}, whose proof is provided in \Cref{sec:amortized-to-worst-case}.
\begin{definition}[Stable Dynamic Uniform Sampling]
    Given a dynamic set of elements (undergoing insertions and deletions) denoted $S_t$ at time $t$, a \emph{stable uniform dynamic sampler} maintains a sample $x_t \in S_t$ (or reports $\emptyset$ if $S_t = \emptyset$) such that:
    \begin{enumerate}
        \item At every time $t$, each $x \in S_t$ is sampled with probability exactly $1/|S_t|$.
        \item The probability that the sampled element changes after an update is at most $1/|S_t|$.
    \end{enumerate}
\end{definition}
\begin{lemma}
    \label{lemma:stable-uniform-dynamic-sampling}
    There exists a data structure for stable dynamic uniform sampling that uses $O(\log n)$ update time per operation under an oblivious adversary.
    The data structure is randomized and succeeds with high probability.
\end{lemma}

Throughout this section denote each edge update by the pair $\set{(u,v),s}$, where $s\in \set{\pm 1}$ indicates whether the edge $(u,v)$ is inserted or deleted.
Note that \Cref{alg:dynamic-star-contraction} excludes from $G'$ all the vertices in $V\setminus R$ that are not contracted to any vertex in $R$.
The proof of \Cref{lemma:dynamic-star-contraction} follows immediately from the following two lemmas, which prove the correctness and time complexity of \Cref{alg:dynamic-star-contraction} respectively.

\begin{algorithm}[htbp]
    \caption{Dynamic Star Contraction}
    \label{alg:dynamic-star-contraction}
    \begin{algorithmic}[1]
    \State \textbf{Input:} An unweighted multigraph $G=(V,E)$, parameter $\tau > 0$
    \State \textbf{Output:} A contracted graph $G'$
    \Procedure{Dynamic-Star-Contraction}{$G,\tau$}
        \State $R\gets $ sample every $v\in V$ with probability $p=800\log n /\tau$
        \State for each $v\in V\setminus R$ let $R_v \gets $ initialize stable uniform dynamic sampler  \Comment{the vertex $r\in R$ to which $v$ is contracted}
        \State for each $v\in R$ let $R_v \gets v$ \Comment{by convention, since $v\in R$ is never contracted}
        \State $G' \gets (R,\emptyset)$
        \For{each edge update $\set{(u,v),s}$}
            \State if $s=1$ then $E\gets E \cup \set{(u,v)}$ else $E\gets E\setminus \set{(u,v)}$
            \State if $R_v,R_u$ are both not null then $w(R_u,R_v) \gets w(R_u,R_v)+s$%
            \label{lst:line:both-in-R-or-both-not}
            \State if $u,v \in R$ or $u,v \not \in R$ then \textbf{continue} \label{lst:line:one-in-R}
            \State if $u\not\in R$ then swap $u,v$ \Comment{$|\set{u,v}\cap R|=1$ by line~\ref{lst:line:one-in-R}}
            \State add/remove $u$ to/from $R_v$ \Comment{update the stable uniform dynamic sampler for $v$}%
            \label{lst:line:update-R-sampler}
            \If{$R_v$ has changed}
                \State $R_v^{\text{old}} \gets $ old $R_v$, $R_v^{\text{new}} \gets $ new $R_v$
                \For{each $w\in R$} \label{lst:line:for-each-w-in-R}
                    \State $E_{vw} \gets \set{x\in V\setminus R \mid (v,x)\in E, R_x=w}$
                    \State if $R_v^{\text{old}}$ is not null then $w(R_v^{\text{old}},w) \gets w(R_v^{\text{old}},w)-|E_{vw}|$
                    \State if $R_v^{\text{new}}$ is not null then $w(R_v^{\text{new}},w) \gets w(R_v^{\text{new}},w)+|E_{vw}|$
                \EndFor
            \EndIf
        \EndFor
    \EndProcedure
    \end{algorithmic}
\end{algorithm}
\begin{lemma}[Correctness]
    \label{lemma:dynamic-star-contraction-correctness}
    Whenever $\delta_G\ge \tau$, the output of \Cref{alg:dynamic-star-contraction} with parameter $\tau$ is a complete contraction and constitutes a $\tau$-star contraction of a given dynamic graph $G$.
\end{lemma}
\begin{proof}
    We begin by showing that if every edge of $G$ is mapped to an edge in $G'$, then the output is a $\tau$-star contraction.
    Notice that the vertices in $R$ are sampled uniformly at random with probability $p=800\log n /\tau$.
    By the properties of the stable uniform dynamic sampling, we have that for every vertex $v\in V\setminus R$, the vertex $R_v$ is sampled uniformly at random from the neighbors of $v$ in $R$.
    Therefore, once every edge of $G$ is mapped to an edge in $G'$, the output is a $\tau$-star contraction.
    We now show that the output of the algorithm is complete.
    Whenever $\tau\ge \delta_G$ by \Cref{theorem:original-star-contraction}, we have that every vertex $v\in V\setminus R$ is contracted with high probability.
    In addition, notice that once every vertex $v\in V\setminus R$ is contracted, all edges of $G$ have been mapped to edges in $G'$.
    Therefore the output of the algorithm is complete.
\end{proof}
\begin{lemma}[Time Complexity]
    \label{lemma:dynamic-star-contraction-time}
    Given an unweighted graph $G=(V,E)$ on $n$ vertices, \Cref{alg:dynamic-star-contraction} with parameter $\tau$ requires amortized time $\tO(1)$ (and worst case time $\tO(n)$) per edge.
    Furthermore, in expectation the algorithm updates at most $O(1)$ edges in $G'$ per edge update in $G$ and in the worst-case it updates $\tO(n/\tau)$ edges in $G'$ per edge update in $G$.
\end{lemma}

\subsection{Amortized Time Complexity}
The proof of \Cref{lemma:dynamic-star-contraction-time} relies on the following claim.
\begin{claim}
    \label{claim:dynamic-star-contraction-change-rate}
    For each edge update $\{(u,v),s\}$, with probability at least $1-800\log n / \max \left\{ d_G(v), \tau \right\}$ both $R_v,R_u$ remain the unchanged.
\end{claim}
\begin{proof}[Proof of Claim \ref{claim:dynamic-star-contraction-change-rate}]
    To bound the probability $R_v$ changes, we may assume without loss of generality that $v\not\in R$ (since if $u,v\in R$ then $R_v,R_u$ do not change).
    Notice that under this assumption, $R_u$ does not change, since $v\not\in R$ and so $R_u$ is not updated.
    Observe that $R_v$ can change only if $u\in R$, hence
    \begin{equation}
        \label{eq:dynamic-star-contraction-change-rate}
        \Probability{R_v \text{ changes}} 
        = 
        \Probability{u\in R} \cdot \Probability{R_v \text{ changes} \mid u\in R}
        = p \cdot\Probability{R_v \text{ changes} \mid u\in R}
        ,
    \end{equation}
    where the last equality is since $R$ is obtained by sampling each vertex independently with probability $p=800\log n /\tau$.
    Hence the probability that $R_v$ changes is at most $800\log n /\tau$.
    We now show a stronger bound on the probability that $R_v$ changes when $d_G(v) \ge \tau$.
    Notice that $\Probability{R_v \text{ changes} \mid u\in R}=1/d_R(v)$ by \Cref{lemma:stable-uniform-dynamic-sampling}, where $d_R(v)$ is the number of neighbors of $v$ in $R$.
    To bound $d_R(v)$, notice that $\Exp{d_R(v) \mid u\in R}{R} \ge \Exp{d_R(v)\in R}{R} = pd_G(v)$, where the inequality is since conditioning on the event that $u\in R$ increases the expected number of neighbors of $v$ in $R$ and the equality is since $R$ is sampled uniformly at random from $V$.
    Therefore, by Chernoff's bound (\Cref{theorem:chernoff}),
    \begin{equation*}
        \Probability{d_R(v) \leq p d_G(v)/2} 
        \leq 2\exp\left(-\frac{pd_G(v)}{12}\right)
        \leq 2\exp\left(-\frac{800\log n}{12}\right)
        \leq n^{-10}
        ,
    \end{equation*}
    where the second inequality is from $d_G(v)\ge \tau$.
    Hence, with probability at least $1-n^{-10}$ we have that $d_R(v) \geq p d_G(v)/2$.
    Plugging this back into \Cref{eq:dynamic-star-contraction-change-rate},
    \begin{equation*}
        \Probability{R_v \text{ changes}} 
        = 
        \Probability{u\in R} \cdot \Probability{R_v \text{ changes} \mid u\in R}
        \leq 
        p \cdot \left( \frac{2}{p d_G(v)}+n^{-10} \right)
        \le 
        \frac{800 \log n}{d_G(v)}
        .
    \end{equation*}
    Therefore, the probability that both $R_v,R_u$ remain unchanged is at least $1-800\log n / \max \left\{ d_G(v), \tau \right\}$, as claimed.
\end{proof}
\begin{proof}[Proof of \Cref{lemma:dynamic-star-contraction-time}]
    Observe that if $R_v,R_u$ do not change, then updating the weight of the edge $(R_u,R_v)$ in line~\ref{lst:line:both-in-R-or-both-not} takes $O(1)$ time and updating the samplers for $u,v$ in line~\ref{lst:line:update-R-sampler} takes $O(\log n)$ time by \Cref{lemma:stable-uniform-dynamic-sampling}.
    They cannot both change since that implies that both $u,v\not\in R$, so assume exactly one of them changes and without loss of generality it is $R_v$.
    In this case, in line~\ref{lst:line:for-each-w-in-R} the algorithm updates in $G'$ the weight of all edges corresponding to edges incident to $v$ in $G$.

    To minimize the number of times we update the edges of $G'$, we go over all edges incident to $v$ in $G$ to find the size of the sets $E_{vw}$, for every $w\in R$.%
    \footnote{Observe that it is simpler to iterate over all the edges incident to $v$ in $G$ and update the weights of the edges in $G'$ during each iteration. However, this creates $d_G(v)$ edge updates in $G'$ which then takes $O(d_G(v) n/\delta_G)=\omega(n)$ time (if $d_G(v)>\delta_G$) to update the tree packing in $G'$. In contrast, our approach only updates the edges of $G'$ at most $\tO(n/\tau)$ times per update in $G$.}
    This requires $O(d_G(v))$ time, since we need to iterate over all edges incident to $v$.
    Then, we update the weights of the edges $(R_v^{\text{old}},w)$ and $(R_v^{\text{new}},w)$ in $G'$ in $O(1)$ time for every $w\in R$, which requires $O(\min\set{|R|,d_G(v)})\le O(d_G(v))$ time overall, since at most $d_G(v)$ sets $E_{uw}$ can be non-empty.
    Therefore, by \Cref{claim:dynamic-star-contraction-change-rate}, the expected time spent on this operation is $O(\log n)$.
    In the worst case each update requires at most $O(d_G(v))$ time, which is $O(n)$ per edge.

    To bound the number of changed edges in $G'$, notice that if $R_v$ does not change then at most one edge is updated in $G'$.
    Otherwise, the number of edges updated in $G'$ is at most $\max{2d_G(v),\tO(n/\tau)}$, since each vertex in $G'$ has at most $\tO(n/\tau)$ neighbors by the size of $R$.
    Combining the above with \Cref{claim:dynamic-star-contraction-change-rate} we find that the expected number of edges updated in $G'$ per edge update in $G$ is $O(1)$.
    By the same bound presented above, the worst-case number of edges whose weight is updated in line~\ref{lst:line:for-each-w-in-R} per edge update in $G$ is $\tO(n/\tau)$.
\end{proof}

\subsection{From Amortized to Worst-Case Time Complexity via Incremental Rebuilding}
\label{sec:amortized-to-worst-case}
In this section we prove \Cref{lemma:dynamic-lazy-star-contraction}, bounding the worst-case time complexity of \Cref{alg:dynamic-star-contraction} with $(n\log^4 n/\delta_G(v))$-incremental rebuilding.
Notice that the incremental rebuilding completes the update of the star contraction $G'$ after $\delta_G/\log^4 n$ edge updates in $G$.
Therefore, the probability that the queue of the algorithm is not empty (and hence the contraction is incomplete) at any given time $t$ is bounded by the probability that some vertex changes its representative $R_v$ in the last $\delta_G/\log^3 n$ edge updates.
The following lemma shows that this probability is $O(1/\log^2 n)$ whenever $\delta_G \ge \log^4 n$.
Otherwise, if $\delta_G<\log^4 n$, then the algorithm finishes processing its queue after a single edge update since $n\log^4 n/\delta_G \ge n$, and its output is complete (as long as $\delta_G\ge\tau$).
\begin{claim}
    \label{claim:dynamic-lazy-star-contraction}
    If at a given time $t$, the minimum degree of $G$ is $\delta>\log^4 n$, then given a subsequent sequence of updates of length $N\eqdef O(\delta/\log^3 n)$  the probability that at least one vertex $v\in V$ changes its center vertex $R_v$ during the sequence is at most $O(1/\log^2 n)$.
\end{claim}
\begin{proof}
    Note that the minimum degree of $G$ after the sequence of updates is at least $(1-o(1))\cdot \delta\ge \delta/2$.
    Let $I_i$ be the indicator for the event that there exists $v\in V\setminus R$ such that $R_v$ changes during the $i$-th edge update.
    Notice that $\Exp{I_i}{} \leq 800\log n/(\delta/2)$ by \Cref{claim:dynamic-star-contraction-change-rate} and that $\delta_G \ge \delta/2$ throughout the sequence of updates.
    Therefore, $\Exp{\sum_{i=1}^{N} I_i}{}\le 1600 N \log n/\max\set{\delta/2,\tau}\le O(1/\log^2 n)$.
    To conclude, by Markov's inequality we have that
    \begin{equation*}
        \Probability{\sum_{i=1}^{N} I_i \geq 1} 
        \leq 
        \frac{\Exp{\sum_{i=1}^{N} I_i}{}}{1} 
        \leq 
        O\left(\frac{1}{\log^2 n}\right)
        .
    \end{equation*}
\end{proof}

\subsection{Stable Dynamic Uniform Sampling}
In this section we prove \Cref{lemma:stable-uniform-dynamic-sampling}, it is based on a standard techniques (see e.g. \cite{Luby86,BDM02,BNR02,DLT07}) and we include it for completeness.
\begin{proof}[Proof of \Cref{lemma:stable-uniform-dynamic-sampling}]
    As each element comes, assign it a priority by uniformly sampling an integer in $[1,n^{10}]$, if at any time two elements have the same priority, the algorithm fails.
    Then, maintain a heap of all the items in the sequence, ordered by their priority.
    When an element is deleted, remove it from the heap.
    The minimum element of the heap is the sampled item.
    It is known that a dynamic heap can be maintained in $O(\log n)$ worst-case time per operation, where $n$ is the number of elements in the heap.

    We now prove the two properties of the sampling.
    Notice that the minimum priority element is distributed uniformly, and hence each element has a probability of $1/|S_t|$ to be sampled at time $t$.
    To show the stability, notice that under a deletion the probability that the element deleted is the minimum element is $1/|S_t|$.
    If we are given an insertion, notice that again by symmetry we have that the probability that the new element is sampled is $1/|S_t|$.
    Therefore, we obtain the desired properties of the sampling.
    Finally, notice that the algorithm fails with probability at most $n^{-10}$.
\end{proof}

{\small
  \bibliographystyle{alphaurl}
  \bibliography{bibliography}
} %

\appendix

\section{Concentration Inequalities}
\label{sec:preliminaries}
\begin{theorem}\label[theorem]{theorem:chernoff}
  Let $X_1,\ldots,X_m\in [0,a]$ be independent random variables.
    For any $\delta \in [0,1]$ and $\mu \geq \mathbb{E}\left[\sum_{i=1}^{m}X_i\right]$, we have
        \begin{align}
            \nonumber
            \mathbb{P}\left[
                \left|
                    \sum_{i=1}^{m}X_i - \mathbb{E}\left[\sum_{i=1}^{m}X_i\right]
                \right|
                \geq \delta \mu
            \right]
            \leq
            2\exp
            \left(
                -\frac{\delta^2\mu}{3a}
            \right)
            .
        \end{align}
\end{theorem}

\end{document}